\documentclass[a4paper,UKenglish,cleveref, autoref, thm-restate]{lipics-v2019}

\usepackage{tikz}
\usetikzlibrary{calc}
\usetikzlibrary{patterns}
\tikzstyle{mybox} = [draw=black, fill=white,  thick,
    rectangle, inner sep=10pt, inner ysep=20pt]
\tikzstyle{mybox} = [draw=black, fill=white,  thick,
    rectangle, inner sep=2pt, inner ysep=2pt]

\title{On guarding polygons with holes} 

\titlerunning{} 

\author{Sharareh Alipour}{Institute for Research in Fundamental Sciences, Iran}{Alipour@ipm.ir}{https://orcid.org/0000-0002-3626-8960}{}

\authorrunning{S. Alipour} 

\Copyright{Sharareh Alipour} 



\category{} 

\relatedversion{} 

\supplement{}

\nolinenumbers 

\EventEditors{John Q. Open and Joan R. Access}
\EventNoEds{2}
\EventLongTitle{42nd Conference on Very Important Topics (CVIT 2016)}
\EventShortTitle{CVIT 2016}
\EventAcronym{CVIT}
\EventYear{2016}
\EventDate{December 24--27, 2016}
\EventLocation{Little Whinging, United Kingdom}
\EventLogo{}
\SeriesVolume{42}
\ArticleNo{23}

\begin{document}

\maketitle

\begin{abstract}
There is an old conjecture by Shermer \cite{sher} that in a polygon with $n$ vertices and $h$ holes, $\lfloor \dfrac{n+h}{3} \rfloor$ vertex guards are sufficient to guard the entire polygon. The conjecture is proved for $h=1$ by Shermer \cite{sher} and Aggarwal \cite{aga} seperately. 
 In this paper, we prove a theorem similar to the Shermer's conjecture for a special case where the goal is to guard the vertices of the polygon (not the entire polygon) which is equivalent to finding a dominating set for the visibility graph of the polygon. Our proof also guarantees that the selected vertex guards also cover the entire outer boundary (outer perimeter of the polygon) as well. 
\end{abstract}

\section{Introduction}

A set $S$ of points is said to guard a polygon if, for every point $p$ in the polygon, there is some $q\in S$ such that the line segment between $p$ and $q$ is inside the polygon.

The art gallery problem asks for the minimum number of guards that are sufficient to guard any polygon with $n$ vertices.
There are numerous variations of the original problem that are also referred to as the art gallery problem. In some versions guards are restricted to the perimeter, or even to the vertices of the polygon which are called vertex guards. Some versions require only the perimeter or a subset of the perimeter to be guarded.
The version in which guards must be placed on vertices and only vertices need to be guarded is equivalent to the minimum dominating set problem for the visibility graph of the polygon.

 In graph theory, for a given graph $G$ with vertex set $V(G)$,  $U\subseteq V(G)$ is a dominating set for $G$ if  every vertex $v \in V(G)\backslash U$ has a neighbor in $U$. Minimum dominating set problem is to find a dominating set $V^*\subseteq V(G)$ such that the size of $V^*$ (denoted by $|V^*|$) is minimum among all dominating sets. 

\subsection*{Related results}
Chv\'{a}tal's art gallery theorem \cite{cha} states that $\lfloor \frac{n}{3}\rfloor$ vertex guards are always sufficient and sometimes necessary to guard a simple polygon with $n$ vertices. Later, Fisk \cite{fisk} gave a short proof for Chv\'{a}tal's art gallery theorem.

O'Rourke \cite{or} proved that any polygon $P$ with $n$ vertices and $h$ holes can be guarded with at most $\lfloor \frac{n+2h}{3}\rfloor$  vertex guards. Note that $n$ is the total number of vertices of the polygon including the boundary and holes.
But Shermer conjectured that any polygon $P$ with $n$ vertices and $h$ holes can always be guarded with $\lfloor \frac{n+h}{3}\rfloor$ vertex guards.
This conjecture has been proved by Shermer \cite{sher} and Aggarwal \cite{aga} independently for $h=1$. For $h>1$, the conjecture is still open for more than 35 years.
However Hoffmann, Kaufmann and Kriegel in  \cite{hof} and  Bjorling-Sachs and Souvaine in \cite{bj} proved Shermer's conjecture for point guards (i.e. the guards can be chosen from any points inside or on the boundary of the polygon).

\subsection*{Our result}

In this paper,  we prove that every polygon with holes has a special kind of triangulation to be specified shortly. Next by using this theorem, we prove that  $\lfloor \frac{n+h}{3}\rfloor$ vertex guards are sufficient to guard the boundary of a polygon with $n$ vertices and $h$ holes.
By boundary of $P$ we mean the outer perimeter of $P$.
 As far as we know this version has not been studied. 

\section{Special triangulation}
In this section, we present some basic definitions and a theorem in order to prove our main result.
It has been proved that every polygon with (or without) holes can be triangulated and this triangulation is not always unique.
\begin{definition}
In a given polygon $P$ with $h$ holes, a triangle $\Delta$ in a triangulation of $P$ is called a special triangle if one of its edges is an edge of a hole and the apex vertex is a vertex of the polygon not on that hole (see Figure \ref{poly1}). 
\end{definition}

\begin{figure}[htpb]
\centering
\begin{tikzpicture}[scale=0.6]
\draw (-2,0) node[right] {\scriptsize$P$};

\draw(0,0)--(5,2);
\draw(5,2)--(7,0);
\draw(7,0)--(9,1);
\draw(9,1)--(10,-3);
\draw(10,-3)--(5,-3);
\draw(5,-3)--(3,-6);
\draw(3,-6)--(0,-4);
\draw(0,-4)--(1,-2);
\draw(1,-2)--(0,0);

\draw(4,-0.5)--(6,-0.5);
\draw(6,-0.5)--(5,-2);
\draw(5,-2)--(4,-0.5);
\draw (5,-1) node {\scriptsize$h_1$};

\draw(1,-4)--(4,-4);
\draw(4,-4)--(2,-5);
\draw(2,-5)--(1,-4);
\draw (2.33,-4.3) node {\scriptsize$h_2$};

\draw (1,-4) node[above] {\scriptsize$e$};
\draw (4,-4) node[above] {\scriptsize$f$};

\draw(5,2)--(4,-0.5)[dashed];
\draw(5,2)--(6,-0.5)[dashed];
\draw (5,0) node {\scriptsize$\Delta$};

\draw (5,2) node[above] {\scriptsize$a$};
\draw (4,-0.5) node[left] {\scriptsize$b$};
\draw (6,-0.5) node[right] {\scriptsize$c$};

\draw(5,-2)node[below]{\scriptsize$d$};

\draw(2.5,-3.5)node[above]{\scriptsize$\Delta'$};

\draw(4,-0.5)--(1,-4)[dashed];
\draw(4,-0.5)--(4,-4)[dashed];

\end{tikzpicture}

\caption{A polygon with $2$ holes. In this example, $\Delta$ and $\Delta'$ are special triangles. 
}
\label{poly1}
\end{figure}
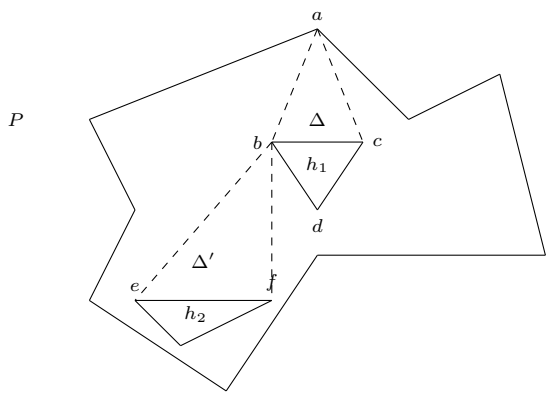
\begin{theorem}
\label{special}
Every polygon with holes has a triangulation with a special triangle.
\end{theorem}
Note that according to the definition of special triangle, one edge of a special triangle is always on a hole and its apex vertex is on the boundary or on a different hole.
In the following we explain the proof of Theorem \ref{special}.

\subsection{Proof of Theorem \ref{special}}
Consider a triangulation of our polygon. If the number of holes is bigger than one, then there must exists an edge of this triangulation from a vertex of one of these holes to a vertex of the boundary of the polygon. If we make this edge into two parallel edges of very small distance from each other, we can make this hole into a part of the boundary, hence reduce the number of holes. A special triangle for this reduced polygon is also a special triangle for the original polygon. By induction on the number of holes, therefore it is enough to consider only the case when we have a polygon with one hole. Let us assume that this polygon with one hole has $n$ vertices, and the existence of a special triangle for a polygon with one hole is proved for the case when the number of vertices is less than $n$. Note that the smallest possible $n$ is $6$ that happens when we have a triangle with a triangle as hole inside. Any triangulation for this particular polygon has a special triangle and in fact 3 special triangles. With this assumption, we can assume that non-adjacent vertices of the boundary can not see each other. Since if they do, the chord connecting them divide the boundary into two smaller parts where the hole is inside one of them. A special triangle for this smaller instance of a polygon with one hole, is also a special triangle for the original polygon. Hence a triangulation for the polygon does not have a triangle whose vertices are all on the boundary. This implies that any vertex of the boundary must see at least one vertex of the hole. Assume that we do not have a special triangle, we want to reach to a contradiction.

Let $B_1$ be a vertex of the boundary. According to the previous argument, it will see a vertex on the hole, say $H_1$. Without loss of generality, we may assume that the ray $B_1H_1$ in a counter clockwise rotational sweep, sees part of the edge $H_1H_1'$ of the hole.

Since a special triangle does not exist, this rotating ray will hit an obstacle that prevents it from seeing the entire edge $H_1H_1'$. Let $H_1''$ be the point on the edge $H_1H_1'$  obtained by rotating this ray until it hits an obstacle.
This obstacle is either a vertex from the boundary or a vertex from the hole. Assume that it is a vertex  $B_2$ from the boundary (we will discuss the second possibility shortly). Since the vertices on the boundary do not see non-adjacent vertices on the boundary, $B_2$ must be adjacent to $B_1$. Then $B_2$ also sees the portion $H_1H_1''$ of $H_1H_1'$ and even more. Repeating the sweeping procedure with the ray $B_2H_1''$, we will hit another obstacle, unless $B_2H_1H_1'$ is a special triangle. The sequence of obstacles obtained this way can not be all the vertices of the boundary, because they keep seeing larger and larger portions of the edge $H_1H_1'$ and hence they are different and we only have a finite number of vertices of the boundary.  So we will reach a vertex $H_2$ of the hole after say $k_1\ge 1$ steps. The vertices $B_1,\dots, B_{k_1}$ are consecutive vertices on the boundary and the angles $B_{i-1}B_iB_{i+1}$ are all less than $\pi$, since they are obtained by counter clockwise sweeps. The vertex $B_{k_1}$ will see a portion of the edge of the hole with the end-point $H_2$, say $H_2H_2'$. Since both of the edges of the hole with end-point $H_2$ are to the right of the ray $B_{k_1}H_2$, therefore we still need to rotate this ray counter clockwise along the edge $H_2H_2'$ and hence if we hit a boundary vertex obstacle, say $B_{k_1+1}$, the angle $B_{k_1-1}B_{k_1}B_{k_1+1}$ is less than $\pi$.

Notice that the obstacles encountered for a vertex of the boundary $B$ by rotating its corresponding ray counter-clockwise, can not all be among the vertices of the hole. In this case the rotating ray will make a full rotation of $2\pi$, and this is impossible since the maximum rotational angle that this ray can have is less than the angle of the vertex $B$, which is definitely less than $2\pi$ degrees.

Now since we are assuming that no special triangle exists, the process of sweeping rays and hitting obstacles will be continued forever. Producing a sequence of vertices of the boundary and holes.  Since after reaching a boundary vertex, we can not get only vertices from the hole, the sequence of boundary vertices that are consecutive vertices of the polygon must come back to the starting point. Hence we  go along all the vertices of the boundary, however the outer angles $B_{i-1}B_iB_{i+1}$ are all less than $\pi$. This is a contradiction.

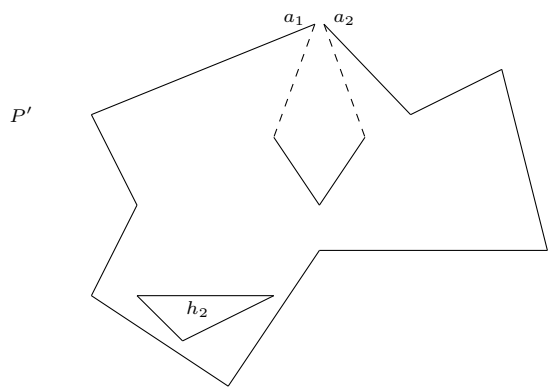
\begin{figure}[htpb]
\centering
\begin{tikzpicture}[scale=0.6]

\draw (-2,0) node[right] {\scriptsize$P'$};

\draw(0,0)--(4.9,2);

\draw(5.1,2)--(7,0);

\draw(7,0)--(9,1);
\draw(9,1)--(10,-3);
\draw(10,-3)--(5,-3);
\draw(5,-3)--(3,-6);
\draw(3,-6)--(0,-4);
\draw(0,-4)--(1,-2);
\draw(1,-2)--(0,0);

\draw(6,-0.5)--(5,-2);
\draw(5,-2)--(4,-0.5);

\draw(1,-4)--(4,-4);
\draw(4,-4)--(2,-5);
\draw(2,-5)--(1,-4);
\draw (2.33,-4.3) node {\scriptsize$h_2$};

\draw(4.9,2)--(4,-0.5)[dashed];
\draw(5.1,2)--(6,-0.5)[dashed];

\draw (4.9,2.1) node[left] {\scriptsize$a_1$};
\draw (5.1,2.1) node[right] {\scriptsize$a_2$};

\end{tikzpicture}

\caption{We split the vertex $a$ into $2$ vertices $a_1$ and $a_2$. Now the polygon has $1$($h-1$) hole and $16$($n+1$) vertices.}
\label{poly2}
\end{figure}

\section{Guarding vertices with vertex guards}
Now as a result of Theorem \ref{special}, we present our main theorem.

\begin{theorem}
\label{main}
For a given polygon $P$ with $n$ vertices and $h$ holes, $\lfloor \frac{n+h}{3}\rfloor$ vertex guards are always sufficient to guard the vertices of $P$
and also the entire boundary.
\end{theorem}
\begin{proof}
The proof is by induction on the number of holes.
Chv\'{a}tal's theorem implies that when  $h=0$, $\lfloor \frac{n}{3}\rfloor$ vertex guards are sufficient to guard the entire polygon.
Suppose that the theorem is proved for $h-1$ holes. Now suppose that we are given a polygon $P$ with $n$ vertices and $h$ holes. According to Theorem \ref{special}, there is a triangulation with a special triangle $\Delta$. Suppose that $\Delta$ has a vertex $a$, on the boundary or a hole and another edge, $bc$ on some other hole.
We split $a$ into two vertices $a_1$ and $a_2$. So, $P$ is changed into a polygon $P'$ with $h-1$ holes and $n+1$ vertices (See Figure \ref{poly2}). According to the induction assumption, we can choose $\lfloor \frac{n+1+h-1}{3}\rfloor$ vertex guards that guard the vertices of polygon $P'$ and the boundary of $P'$. Since $a_1$ and $a_2$, are guarded in $P'$, then all vertices of $P$ are guarded by at most $\lfloor \frac{n+h}{3}\rfloor$ vertex guards of $P$. Also by induction the boundary of $P'$ is guarded. On the other hand the  boundary of $P$ is a subset of the boundary of $P'$, so the  boundary of $P$ is guarded too. Note that if we have two vertex guards on $a_1$ and $a_2$, in $P$ they are combined into one vertex guard.
\end{proof}
\begin{remark}
Note that in fact this proof, gives something slightly more. The guards will cover not only the entire outer perimeter of the polygon, but also the perimeter of holes with an exception of at most $h$ segments on them. The reason is that the segment that is the base of the special triangle that was used in the proof above is not necessarily guarded, so a similar induction on the number of holes will prove our claim.
\end{remark}

\section{Concluding remarks}
In this paper, we have proved a theorem similar to the Shermer's conjecture for a special case of the Art Gallery problem where we only need to guard the vertices of the polygon. The proof is based on the existence of a special triangle in the polygon. 
The proposed  algorithm is simple and easy to implement. In future work one can possibly extend this idea for the general case. 

Also for a given connected graph $G$, it has been proved that the size of minimum dominating set of $G$ is $\leq \dfrac{n}{2}$. So if we construct a polygon with minimum number of holes such that its visibility graph is isomorphic to $G$, our proof gives an upper bound of $\lfloor \frac{n+h}{3}\rfloor$ for the size of minimum dominating set of $G$.



\bibliography{visibility}

\end{document}